\documentclass[english]{article}
\usepackage{mathpazo}
\usepackage[T1]{fontenc}
\usepackage[latin9]{inputenc}
\usepackage{babel}
\usepackage{latexsym}
\usepackage{prettyref}
\usepackage{booktabs}
\usepackage{mathrsfs}
\usepackage{mathtools}
\usepackage{amsmath}
\usepackage{amsthm}
\usepackage{amssymb}
\usepackage{graphicx}
\usepackage{setspace}
\usepackage[unicode=true,pdfusetitle,
 bookmarks=true,bookmarksnumbered=false,bookmarksopen=false,
 breaklinks=false,pdfborder={0 0 1},backref=false,colorlinks=false]
 {hyperref}

\makeatletter

\providecommand{\tabularnewline}{\\}

  \theoremstyle{plain}
  \newtheorem*{assumption*}{\protect\assumptionname}
  \theoremstyle{remark}
  \newtheorem*{rem*}{\protect\remarkname}
\theoremstyle{plain}
\newtheorem{thm}{\protect\theoremname}
  \theoremstyle{plain}
  \newtheorem{lem}[thm]{\protect\lemmaname}
  \theoremstyle{plain}
  \newtheorem{prop}[thm]{\protect\propositionname}

\pdfoutput=1

\usepackage{clrscode3e}

\usepackage{microtype}

\usepackage{geometry}
\usepackage[dvipsnames]{xcolor}
\newcommand\myshade{100}
\definecolor{mylinkcolorhtml}{HTML}{0066cc}
\definecolor{mycitecolorhtml}{HTML}{cc6600}
\definecolor{myurlcolorhtml}{HTML}{0066cc}
\colorlet{mylinkcolor}{mylinkcolorhtml}
\colorlet{mycitecolor}{mycitecolorhtml}
\colorlet{myurlcolor}{myurlcolorhtml}
\hypersetup{
  linkcolor  = mylinkcolor!\myshade!black,
  citecolor  = mycitecolor!\myshade!black,
  urlcolor   = myurlcolor!\myshade!black,
  colorlinks = true,
}

\newrefformat{assu}{Assumption \ref{#1}}
\newrefformat{prop}{Proposition \ref{#1}}
\newrefformat{fig}{Figure \ref{#1}}
\newrefformat{tab}{Table \ref{#1}}
\newrefformat{def}{Definition \ref{#1}}
\newrefformat{exa}{Example \ref{#1}}
\newrefformat{cor}{Corollary \ref{#1}}
\newrefformat{sec}{\S\ref{#1}}
\newrefformat{subsec}{\S\ref{#1}}
\newrefformat{app}{Appendix \ref{#1}}
\newrefformat{rem}{Remark \ref{#1}}

\renewcommand*{\leq}{\leqslant}
\renewcommand*{\geq}{\geqslant}

\usepackage{upref}

\@ifundefined{showcaptionsetup}{}{%
 \PassOptionsToPackage{caption=false}{subfig}}
\usepackage{subfig}
\AtBeginDocument{
  
}

\makeatother

  \providecommand{\assumptionname}{Assumption}
  \providecommand{\lemmaname}{Lemma}
  \providecommand{\propositionname}{Proposition}
  \providecommand{\remarkname}{Remark}
\providecommand{\theoremname}{Theorem}

\begin{document}
\global\long\def\leq{\leqslant}

\global\long\def\geq{\geqslant}

\global\long\def\card{\operatorname{card}}

\global\long\def\diag{\operatorname{diag}}

\date{\date{}}

\title{Stock loans with liquidation\\
\large A technical report commissioned by the Bank of Nova Scotia\thanks{The views expressed herein are solely those of the author and not
those of any other person or entity.}}

\author{Parsiad Azimzadeh\thanks{David R. Cheriton School of Computer Science, University of Waterloo,
Waterloo ON, Canada N2L 3G1 {\tt \href{mailto:pazimzad@uwaterloo.ca}{pazimzad@uwaterloo.ca}}.} }
\maketitle
\begin{abstract}
We derive a ``semi-analytic'' solution for a stock loan in which
the lender forces liquidation when the loan-to-collateral ratio drops
beneath a certain threshold. We use this to study the sensitivity
of the contract to model parameters.
\end{abstract}

\section{Introduction}

We  study a \emph{stock loan} contract in which the lender can force
the client to liquidate. It was originally pointed out in \cite{xia2007stock}
that stock loans are essentially American call options with negative
interest rates. The negative rate appears since the contract is effectively
discounted by $r-\gamma$, where $\gamma$, the loan interest rate,
is generally larger than the riskless rate of return $r$.

Concretely, a stock loan is a loan of size $q$ obtained from a financial
firm (lender) by posting shares of an asset valued at $s$ as collateral.
Such loans are usually \emph{nonrecourse} in that if the stock price
drops, the borrower (client) may simply \emph{forfeit ownership of
the shares} in lieu of repaying the loan. On the other hand, if the
stock price rises, the client can regain their shares by \emph{repaying
the loan }(along with the accrued interest).

The liquidation clause is useful from the perspective of the lender
as it reduces the amount of risk the lender is exposed to, simultaneously
reducing the rational premium charged for the loan. We find that such
loans are riskless in the absence of jumps, in which case neither
lender nor client benefits from entering into such a loan (simultaneously
motivating the need for a jump-diffusion model; in particular, we
have chosen the HEM for its analytical tractability and ample freedom
in shaping return distributions).

\cite{xia2007stock} studies a perpetual contract (i.e., one that
does not expire) with dividends paid out to the lender. In this original
model, the lender cannot take action once the contract is initiated.
\cite{zhang2009valuation} studies stock loans under regime-switching.
Optimal strategies for both perpetual and finite-maturity contracts
subject to various dividend distribution schemes are studied in \cite{dai2011optimal}.
A model in which the asset is driven by a hyper-exponential jump-diffusion
is considered in \cite{cai2014valuation} for both perpetual and finite-maturity
contracts. Some other works on stock loans are \cite{liang2010stock,liu2010capped,liang2012variational,wong2012stochastic,grasselli2013stock,pascucci2013mathematical,wong2013valuation,chen2015stock,leung2015optimal,lu2015semi}.

\section{\label{sec:hyper_exponential_model}Hyper-exponential model}

Consider a stochastic process $(X_{t}^{x})_{t\geq0}$ following a
hyper-exponential model (HEM)
\[
X_{t}^{x}\coloneqq x+\mu t+\sigma W_{t}+\sum_{i=1}^{N_{t}}Y_{i}
\]
where $(N_{t})_{t\geq0}$ is a (right-continuous) Poisson process
with rate $\lambda$, $(W_{t})_{t\geq0}$ is a standard Brownian motion,
and $(Y_{i})_{i=1}^{\infty}$ is a sequence of i.i.d. hyper-exponential
random variables with p.d.f.
\[
x\mapsto\sum_{i=1}^{m}p_{i}\eta_{i}e^{-\eta_{i}x}1_{\left\{ x\geq0\right\} }+\sum_{j=1}^{n}q_{j}\theta_{j}e^{\theta_{j}x}1_{\left\{ x<0\right\} }.
\]
We make the following assumption throughout:
\begin{assumption*}
$\lambda\geq0$, $p_{i}>0$ for all $i$, $q_{j}>0$ for all $j$,
$1<\eta_{1}<\cdots<\eta_{m}$, $0<\theta_{1}<\cdots<\theta_{n}$,
$m$ and $n$ are nonnegative integers not both equal to zero, and
$\sum_{i=1}^{m}p_{i}+\sum_{j=1}^{n}q_{j}=1$ (subject to the convention
$\sum_{i=1}^{0}\cdot=0$).
\end{assumption*}
The Lévy exponent of the process is
\[
G(x)\coloneqq\sigma^{2}x^{2}/2+\mu x+\lambda\left(\sum_{i=1}^{m}\frac{p_{i}\eta_{i}}{\eta_{i}-x}+\sum_{j=1}^{n}\frac{q_{j}\theta_{j}}{\theta_{j}+x}-1\right)\text{ for }x\in(-\theta_{1},\eta_{1}).
\]
In \cite[Lemma 2.1 and Remark 2.3]{cai2009first}, it is shown that
for any $\alpha\geq\mathbb{M}(G)$ where
\[
\mathbb{M}(G)\coloneqq\inf\left\{ G(x)\colon x\in(-\theta_{1},\eta_{1})\right\} \leq0,
\]
$x\mapsto G(x)-\alpha$ has $m+n+2$ real roots $\beta_{1,\alpha}$,
$\ldots$, $\beta_{m+1,\alpha}$, $-\gamma_{1,\alpha}$, $\ldots$,
and $-\gamma_{n+1,\alpha}$ satisfying (omitting the subscript $\alpha$)
\begin{gather}
-\infty<-\gamma_{n+1}<-\theta_{n}<-\gamma_{n}<\cdots<-\gamma_{2}<-\theta_{1}<-\gamma_{1}\nonumber \\
\leq\beta_{1}<\eta_{1}<\beta_{2}<\cdots<\beta_{m}<\eta_{m}<\beta_{m+1}<\infty.\label{eq:root_inequalities}
\end{gather}
$-\gamma_{1}=\beta_{1}$ is a possibility, in which case there are
only $m+n+1$ \emph{distinct} roots.

\section{Stock loans with liquidation}

We consider a stock loan contract similar to \cite{xia2007stock}.
At the initial time, the client borrows an amount $q$ from the lender
using one share of the stock with initial price $e^{x}$ as collateral.
The stock follows (in log space) the process $X$ defined in \prettyref{sec:hyper_exponential_model}
with $\mu\coloneqq r-\delta-\sigma^{2}/2-\lambda\zeta$ where $r$
is the risk-free rate, $\delta$ is the dividend rate, $\sigma$ is
the volatility, and $\zeta\coloneqq\mathbb{E}e^{Y_{1}}-1$. The loan
is continuously compounded at the rate $\gamma$. At any (stopping)
time $t$, the client can choose to redeem the stock and pay back
$qe^{\gamma t}$. If the loan-to-collateral ratio $qe^{\gamma t-X_{t}^{x}}$
exceeds the level $d$, the lender liquidates the loan, forcing the
client to pay back $qe^{\gamma t}$. During the collateral period,
stock dividends are collected by the lender.
\begin{rem*}
Transfer-of-title stock loans were shut down by the SEC and IRS between
2007-2012 and reclassified as fully taxable sales at inception (see
also \cite{finra}). It stands to reason that the case of $\delta>0$
implies a transfer-of-title (since the lender receives the dividends),
and hence is subject to tax considerations. The case of $\delta=0$
corresponds to dividends being immediately reinvested in the stock
and returned to the client upon redemption, and is thus arguably the
more relevant case for the current era. Other dividend distribution
schemes are visited in \cite{dai2011optimal}.
\end{rem*}
If we assume that the client is able to pick a time to redeem the
stock from $\mathscr{T}$, the set of $[0,+\infty]$ stopping times,
the value of this contract is
\[
\sup_{\tau\in\mathscr{T}}\mathbb{E}\left[e^{-r(\tau\wedge\pi)}\left(e^{X_{\tau\wedge\pi}^{x}}-qe^{\gamma(\tau\wedge\pi)}\right)^{+}\right]\text{ where }\pi\coloneqq\inf\left\{ t\geq0\colon qe^{\gamma t-X_{t}^{x}}\geq d\right\} .
\]
It is understood that the expression in the expectation is zero whenever
$\tau\wedge\pi=+\infty$. However, due to the presence of the stopping
time $\pi$, it is not a simple matter to show that the above is a
free boundary problem with respect to $x$. This makes the above problem
difficult if we are seeking analytical solutions (see, e.g., \cite[Proposition 3.1]{xia2007stock}).
Instead, we consider the simpler
\begin{multline*}
v(x)\coloneqq\sup_{u\in(0,d)}\mathbb{E}\left[e^{-r(\tau\wedge\pi)}\left(e^{X_{\tau\wedge\pi}^{x}}-qe^{\gamma(\tau\wedge\pi)}\right)^{+}\right]\\
\text{ where }\pi\coloneqq\inf\left\{ t\geq0\colon qe^{\gamma t-X_{t}^{x}}\geq d\right\} \text{ and }\tau\coloneqq\inf\left\{ t\geq0\colon qe^{\gamma t-X_{t}^{x}}\leq u\right\} ,
\end{multline*}
in which the client picks instead a level $u$ to stop at based on
the loan-to-collateral ratio.
\begin{assumption*}
$0<d\leq1$, $\delta\geq0$, and $\gamma\geq r\geq0$.
\end{assumption*}
The following is a trivial consequence of the definition of $v$. 
\begin{lem}
\label{lem:bounds_on_v}$(e^{x}-q)^{+}\leq v(x)\leq e^{x}$ everywhere
and $v(x)=(e^{x}-q)$ for $x\leq\ln(q/d)$.
\end{lem}
The following is a trivial, but interesting aside: it establishes
that the client has no reason to take out a stock loan in the absence
of (downward) jumps in the collateral value (equivalently, the lender
is exposed to no risk).
\begin{lem}
\label{lem:no_jumps_implies_riskless}If $n=0$ or $\lambda=0$, then
$v(x)=(e^{x}-q)^{+}$ everywhere.
\end{lem}
\begin{proof}
Note that in this case, $qe^{\gamma\pi(\omega)-X_{\pi(\omega)}^{x}}=d$
for $\mathbb{P}$-almost all $\omega$ such that $\pi(\omega)<\infty$.
Therefore, for any stopping time $\tau$, $qe^{\gamma[\tau(\omega)\wedge\pi(\omega)]-X_{\tau(\omega)\wedge\pi(\omega)}^{x}}\leq d\leq1$
and hence $qe^{\gamma[\tau(\omega)\wedge\pi(\omega)]}\leq e^{X_{\tau(\omega)\wedge\pi(\omega)}^{x}}$
for $\mathbb{P}$-almost all $\omega$ such that $\tau(\omega)\wedge\pi(\omega)<\infty$.
It follows that
\[
v(x)\leq\sup_{\tau\in\mathscr{T}}\mathbb{E}\left[e^{-r(\tau\wedge\pi)}\left(e^{X_{\tau\wedge\pi}^{x}}-qe^{\gamma(\tau\wedge\pi)}\right)\right]\leq\sup_{t\geq0}\left\{ e^{x-\delta t}-qe^{(\gamma-r)t}\right\} =e^{x}-q\text{ for }x\geq\ln q.\qedhere
\]
\end{proof}
Returning to our objective, define a drift-adjusted process $(\tilde{X}_{t}^{x})_{t\geq0}$
by $\tilde{X}_{t}^{x}\coloneqq X_{t}^{x}-\gamma t$ with Lévy exponent
$\tilde{G}(x)\coloneqq G(x)-\gamma x$. Letting $f(x)\coloneqq(e^{x}-q)^{+}$,
we can write 
\begin{multline}
v(x)=\sup_{u\in(0,d)}\mathbb{E}\left[e^{(\gamma-r)\tilde{\tau}}f(\tilde{X}_{\tilde{\tau}}^{x})\right]\text{ where }\tilde{\tau}\coloneqq\inf\left\{ t\geq0\colon\tilde{X}_{t}^{x}\notin(h,H)\right\} \text{,}\\
h\coloneqq\ln(q/d)\text{, and }H\coloneqq\ln(q/u).\label{eq:simple_v}
\end{multline}
For fixed values of $h$ and $H$, we may compute the expectation
\begin{equation}
\mathbb{E}\left[e^{(\gamma-r)\tilde{\tau}}f(\tilde{X}_{\tilde{\tau}}^{x})\right]\label{eq:expectation}
\end{equation}
using \prettyref{prop:cai_theorem3.1} of the appendix (see, in particular,
expression \eqref{eq:closed_form}). Before we can do so, we must
ensure the finitude of the expectation.
\begin{assumption*}
Either (i) $\delta>0$ or (ii) $\delta=0$ and $\frac{d\tilde{G}}{dx}(1)<0$. 
\end{assumption*}
Subject to the above, \cite[Theorem 3.1]{cai2014valuation} holds,
repeated below for convenience.
\begin{prop}
\label{lem:cai_theorem3.1}$\mathbb{E}[\sup_{t\geq0}e^{(\gamma-r)t}f(\tilde{X}_{t}^{x})]<\infty$.
\end{prop}
Some of the computations required to apply \prettyref{prop:cai_theorem3.1}
to \eqref{eq:expectation} are summarized below.
\begin{lem}
\label{lem:f_values}Let $q>0$, $\ln q\leq h<H$, and $f(x)\coloneqq(e^{x}-q)^{+}$.
Then, 
\begin{align*}
f_{0}^{u} & =e^{H}-q, & f_{i}^{u} & =\frac{e^{H}}{\eta_{i}-1}-\frac{q}{\eta_{i}} & \text{for }1\leq i\leq m;\\
f_{0}^{d} & =e^{h}-q, & f_{i}^{d} & =\frac{e^{-h\theta_{j}}q^{1+\theta_{j}}}{\theta_{j}\left(1+\theta_{j}\right)}+\frac{e^{h}}{1+\theta_{j}}-\frac{q}{\theta_{j}} & \text{for }1\leq j\leq n;
\end{align*}
where $f_{i}^{u}$ and $f_{j}^{d}$ are defined in \prettyref{prop:cai_theorem3.1}.
\end{lem}

\section{Numerical results}

\subsection{Algorithm}

Using \prettyref{prop:cai_theorem3.1} of the appendix, a direct computation
shows that the expectation in \eqref{eq:simple_v} is continuous as
a function of $u$. This inspires the algorithm below for approximating
$v$ at $x$. We write $\tilde{\tau}(u)$ to stress the dependence
of $\tilde{\tau}$ (defined in \eqref{eq:simple_v}) on $u$. Below,
$f$ is given by $f(x)\coloneqq(e^{x}-q)^{+}$.

\begin{codebox}

\Procname{\proc{Stock-Loan}$(x;N)$}

\li $V\coloneqq f(x)$

\li \If $x>\ln(q/d)$ \li \Then

    \For $j=1,2,\ldots,N$ \li \Do

           $u_{j}\coloneqq jd/(N+1)$

       \li $\mathrlap{V}\phantom{u_{j}}\coloneqq\max\{\mathbb{E}[e^{(\gamma-r)\tilde{\tau}(u_{j})}f(\tilde{X}_{\tilde{\tau}(u_{j})}^{x})],V\}$

\End

\End

\li \Return $V$

\end{codebox}

The integer $N\geq1$ controls the accuracy of the algorithm. The
expectation is computed using \prettyref{prop:cai_theorem3.1} and
\prettyref{lem:f_values}. There are various modifications one can
make to speed up this algorithm (e.g., using a non-uniform grid $\{u_{j}\}$),
though we do not visit them here.

\subsection{Lender's perspective}

To aid our understanding, we also consider the contract from the lender's
perspective. In particular, let $u$ be the map satisfying $v(x)=x-u(x)$.
This equation has the interpretation that the client retains ownership
of the stock (initially valued at $e^{x}$) and is short the contract
$u$. Conversely, the lender is long the contract $u$. Simple algebra
along with the fact $\mathbb{E}[e^{X_{t}^{x}-rt}]=e^{x-\delta t}$
reveals
\begin{equation}
u(s)=\inf_{u\in(0,d)}\mathbb{E}\left[e^{-r\tau}\left(e^{X_{\tau\wedge\pi}^{x}}\left(e^{\delta(\tau\wedge\pi)}-1\right)+\min\left\{ e^{X_{\tau\wedge\pi}^{x}},qe^{\gamma(\tau\wedge\pi)}\right\} \right)\right]\label{eq:lenders_contract}
\end{equation}
From the form \eqref{eq:lenders_contract}, it is easy to see that
$u$ includes dividend flows to the lender (the case in which dividends
are immediately reinvested in the stock and returned to the client
upon prepayment is equivalent to taking $\delta=0$; we refer to \cite[Section 3]{dai2011optimal}
for an explanation).

\subsection{Jump risk}

\prettyref{fig:varying_arrival_rates} shows the stock loan value
under varying jump arrival rates $\lambda$. As stipulated by \prettyref{lem:no_jumps_implies_riskless},
in the absence of jumps, the loan is riskless (i.e. $v(x)=(e^{x}-q)^{+}$,
or equivalently, $u(x)=\min\{e^{x},q\}$). However, as $\lambda$
increases, the lender is exposed to more risk. This coincides with
our intuition: as the probability of a downward jump is increased,
so too is the probability that $X_{t}^{x}$ jumps below the loan value
$q$.

The lender is exposed to the most risk at a point between $x=\ln(q/d)$
and $x\rightarrow\infty$. This is explained as follows:
\begin{itemize}
\item If the collateral is very close to $\ln(q/d)$, the event that a downward
jump brings $X_{t}^{x}$ below the loan value before the lender is
able to liquidate is unlikely.
\item If the collateral is very large and the client chooses to prepay,
the lender can retrieve the loan value in its entirety.
\end{itemize}
\begin{table}
\begin{centering}
{\scriptsize%
\begin{tabular}{rcl}
\toprule 
\multicolumn{2}{r}{\textbf{Parameter}} & \textbf{Value}\tabularnewline
\midrule 
Model &  & Double exponential ($m=n=1$)\tabularnewline
\midrule 
Risk-free rate & $r$ & $0.05$\tabularnewline
\midrule 
Dividend rate & $\delta$ & $0.02$\tabularnewline
\midrule 
Volatility & $\sigma$ & $0.15$\tabularnewline
\midrule 
Loan interest rate & $\gamma$ & $0.07$\tabularnewline
\midrule 
Initial collateral value & $e^{x}$ & $100$\tabularnewline
\midrule 
Loan value & $q$ & $80$\tabularnewline
\midrule 
Liquidation ratio & $d$ & $80/90\approx90\%$\tabularnewline
\midrule 
Jump arrival rate & $\lambda$ & $0.5$\tabularnewline
\midrule 
Mean up-jump scaling factor  & $1/\eta_{1}$ & $1/2.3$\tabularnewline
\midrule 
Mean down-jump scaling factor & $1/\theta_{1}$ & $1/1.8$\tabularnewline
\midrule 
Up-jump probability & $p_{1}$ & $0.09$\tabularnewline
\midrule 
Down-jump probability & $q_{1}$ & $1-p_{1}=0.91$\tabularnewline
\bottomrule
\end{tabular}}
\par\end{centering}
\caption{Default parameters\label{tab:default_parameters}}
\end{table}

\begin{figure}
\begin{centering}
\subfloat[]{\begin{centering}
\includegraphics[width=4in]{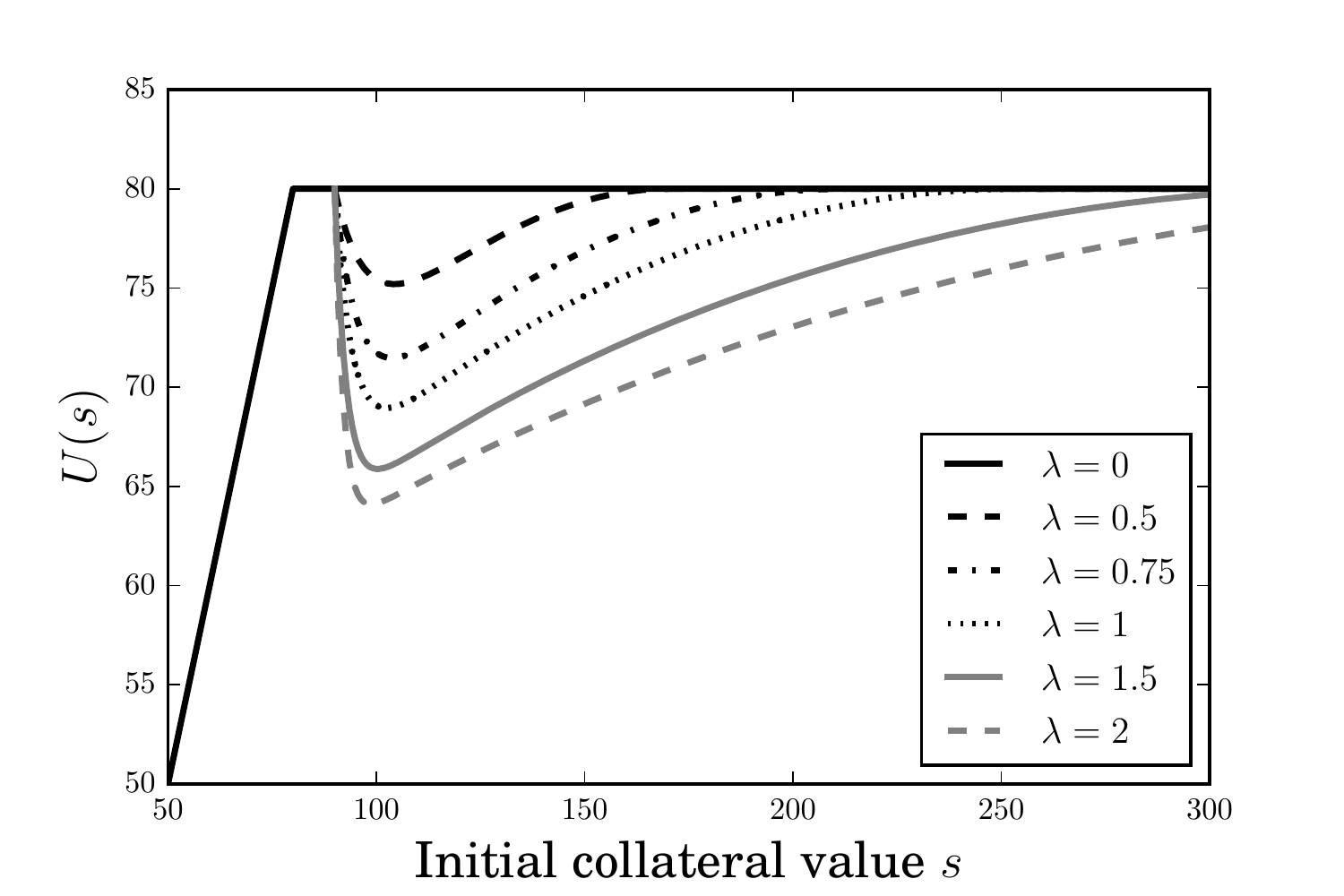}
\par\end{centering}
}
\par\end{centering}
\begin{centering}
\subfloat[$C^{1}$ discontinuity at $e^{x}=q/d=90$.]{\centering{}\includegraphics[width=4in]{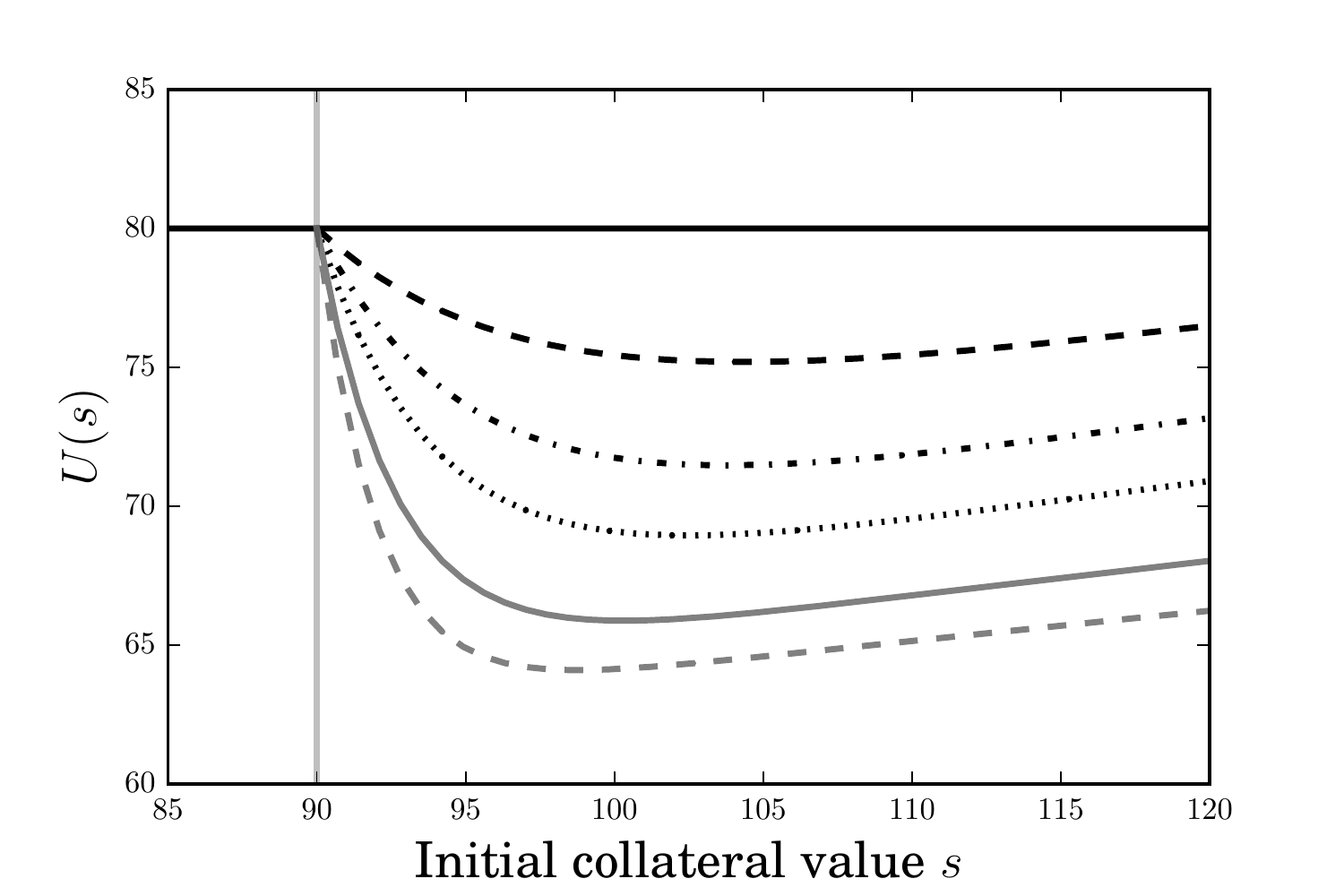}}
\par\end{centering}
\caption{Effect of varying jump arrival rates $\lambda$.\label{fig:varying_arrival_rates}}
\end{figure}

\subsection{Rational values}

Denote by $c$ an up-front premium. At time zero, the client essentially
exchanges the amount $e^{x}-q+c$ for a stock loan \cite{xia2007stock}.
It follows that $\gamma$ and $c$ are required to satisfy the following
identity to preclude arbitrage:
\begin{equation}
v(x;\gamma)=e^{x}-q+c\text{ (equivalently, }u(x;\gamma)=q-c\text{)}\label{eq:rational_identity}
\end{equation}
\prettyref{tab:rational_values} computes some \emph{rational values}
of $\gamma$, and $c$ satisfying \eqref{eq:rational_identity}.
\begin{table}
\begin{centering}
\subfloat[$\lambda=1$]{\begin{centering}
{\scriptsize%
\begin{tabular}{rllllllll}
\toprule 
$q$ & 30 & 40 & 50 & 60 & 70 & 80 & 90 & 100\tabularnewline
\midrule 
$u(x)$ & 30 & 39.29 & 47.26 & 54.51 & 61.35 & 69.09 & 90 & 100\tabularnewline
\midrule 
$c$ & 0 & 0.71 & 2.74 & 5.49 & 8.65 & 10.91 & 0 & 0\tabularnewline
\bottomrule
\end{tabular}}
\par\end{centering}
}
\par\end{centering}
\begin{centering}
\subfloat[$\lambda=2$]{\begin{centering}
{\scriptsize%
\begin{tabular}{rllllllll}
\toprule 
$q$ & 30 & 40 & 50 & 60 & 70 & 80 & 90 & 100\tabularnewline
\midrule 
$u(x)$ & 28.79 & 36.53 & 43.77 & 50.70 & 57.42 & 64.14 & 90 & 100\tabularnewline
\midrule 
$c$ & 1.21 & 3.47 & 6.23 & 9.30 & 12.58 & 15.86 & 0 & 0\tabularnewline
\bottomrule
\end{tabular}}
\par\end{centering}
}
\par\end{centering}
\caption{Rational values of $c$ for $\gamma=0.07$ and $e^{x}=100$ fixed.\label{tab:rational_values}}
\end{table}

Note that for $x\leq\ln(q/d)$, the rational premium is $c=(q-e^{x})^{+}$
(see \prettyref{fig:varying_arrival_rates} for further intuition).
In this case, neither client nor lender has any reason to enter into
the contract. For $x\geq k$ where $k\coloneqq\inf\{x>\ln(q/d)\colon u(x)\geq q\}$,
the rational premium is $c=0$, as the contract is exercised immediately.
It stands to reason that the only region of interest is $\ln(q/d)<x<k$.
Parameterizations in which $k=\ln(q/d)$ (i.e. the free boundary is
``collapsed'') are, therefore, uninteresting.

In particular, note the $C^{1}$ discontinuity occurring at $x=\ln(q/d)$
(\prettyref{fig:varying_arrival_rates}), corresponding to liquidation.
In the event that the free boundary is collapsed, $u(x)=\min\{e^{x},q\}$,
and this discontinuity is removed.
\begin{description}
\item [{Acknowledgements}] The author thanks Shuqing Ma of the Bank of
Nova Scotia and Peter Forsyth, Kenneth Vetzal, and George Labahn of
the University of Waterloo.
\end{description}
\urlstyle{same}

\appendix

\section{Laplace transform of first passage time to two barriers}

We require an expression for the Laplace transform of a process to
two flat barriers $h$ and $H$ with $h<H$. Formally, let
\[
\tau\coloneqq\inf\left\{ t\geq0\colon X_{t}^{x}\notin(h,H)\right\} .
\]
A generalization of the Laplace transform of $\tau$ is studied in
\cite{cai2009pricing}. A trivial modification of the authors' result
is repeated below.
\begin{prop}
\label{prop:cai_theorem3.1}Let $f$ be a nonnegative measurable function
such that 
\[
\int_{0}^{\infty}f(y+H)e^{-\eta_{i}y}dy\text{ and }\int_{-\infty}^{0}f(y+h)e^{\theta_{j}y}dy
\]
are integrable for all $1\leq i\leq m$ and $1\leq j\leq n$. Let
$\alpha\geq\mathbb{M}(G)$ and $x\in(h,H)$. Let $\boldsymbol{N}$
be an $(m+n+2)\times(m+n+2)$ matrix given by
\begin{equation}
\boldsymbol{N}\coloneqq\left[\begin{array}{cccccc}
1 & \cdots & 1 & \overline{x}^{\gamma_{1}} & \cdots & \overline{x}^{\gamma_{n+1}}\\
{\displaystyle \frac{1}{\eta_{1}-\beta_{1}}} & \cdots & {\displaystyle \frac{1}{\eta_{1}-\beta_{m+1}}} & {\displaystyle \frac{\overline{x}^{\gamma_{1}}}{\eta_{1}+\gamma_{1}}} & \cdots & {\displaystyle \frac{\overline{x}^{\gamma_{n+1}}}{\eta_{1}+\gamma_{n+1}}}\\
\vdots & \ddots & \vdots & \vdots & \ddots & \vdots\\
{\displaystyle \frac{1}{\eta_{m}-\beta_{1}}} & \cdots & {\displaystyle \frac{1}{\eta_{m}-\beta_{m+1}}} & {\displaystyle \frac{\overline{x}^{\gamma_{1}}}{\eta_{m}+\gamma_{1}}} & \cdots & {\displaystyle \frac{\overline{x}^{\gamma_{n+1}}}{\eta_{m}+\gamma_{n+1}}}\\
\overline{x}^{\beta_{1}} & \cdots & \overline{x}^{\beta_{m+1}} & 1 & \cdots & 1\\
{\displaystyle \frac{\overline{x}^{\beta_{1}}}{\theta_{1}+\beta_{1}}} & \cdots & {\displaystyle \frac{\overline{x}^{\beta_{m+1}}}{\theta_{1}+\beta_{m+1}}} & {\displaystyle \frac{1}{\theta_{1}-\gamma_{1}}} & \cdots & {\displaystyle \frac{1}{\theta_{1}-\gamma_{n+1}}}\\
\vdots & \ddots & \vdots & \vdots & \ddots & \vdots\\
{\displaystyle \frac{\overline{x}^{\beta_{1}}}{\theta_{n}+\beta_{1}}} & \cdots & {\displaystyle \frac{\overline{x}^{\beta_{m+1}}}{\theta_{n}+\beta_{m+1}}} & {\displaystyle \frac{1}{\theta_{n}-\gamma_{1}}} & \cdots & {\displaystyle \frac{1}{\theta_{n}-\gamma_{n+1}}}
\end{array}\right]\label{eq:matrix}
\end{equation}
where $\beta_{1},\ldots,\beta_{m+1}$ and $\gamma_{1},\ldots,\gamma_{n+1}$
are the real roots of $x\mapsto G(x)-\alpha$ satisfying \eqref{eq:root_inequalities}
and $\overline{x}\coloneqq e^{h-H}$. If $\boldsymbol{N}$ is nonsingular,
\begin{equation}
\mathbb{E}\left[e^{-\alpha\tau}f(X_{\tau}^{x})\right]=\boldsymbol{\varpi}(x)\boldsymbol{N}^{-1}\boldsymbol{f}\label{eq:closed_form}
\end{equation}
where $\boldsymbol{\varpi}(x)$ is a row vector defined as
\[
\boldsymbol{\varpi}(x)\coloneqq\left(e^{\beta_{1}\left(x-H\right)},\,\ldots,\,e^{\beta_{m+1}\left(x-H\right)},\,e^{-\gamma_{1}\left(x-h\right)},\,\ldots,\,e^{-\gamma_{n+1}\left(x-h\right)}\right),
\]
and $\boldsymbol{f}$ is a column vector such that $\boldsymbol{f}=(f_{0}^{u},\,\ldots,\,f_{m}^{u},\,f_{0}^{d},\,\ldots,\,f_{n}^{d})^{\top}$
where 
\begin{align*}
f_{0}^{u} & =f(H), & f_{i}^{u} & =\int_{0}^{\infty}f(y+H)e^{-\eta_{i}y}dy & \text{for }1\leq i\leq m;\\
f_{0}^{d} & =f(h), & f_{j}^{d} & =\int_{-\infty}^{0}f(y+h)e^{\theta_{j}y}dy & \text{for }1\leq j\leq n.
\end{align*}
\end{prop}

\bibliographystyle{plain}
\bibliography{stock_loans_with_liquidation}

\end{document}